\documentclass[a4paper,10pt]{article}
\usepackage[utf8]{inputenc}
\usepackage{amsmath,amsthm,amssymb}
\usepackage{fullpage}
\usepackage{hyperref}
\usepackage{authblk}
\usepackage{amssymb} 
\usepackage{verbatim} 
\usepackage{todonotes}
\usepackage{xspace}
\usepackage{color}
\usepackage{enumerate}
\usepackage{tikz}
\usepackage{subcaption}
\usepackage[]{algorithm}
\usepackage[noend]{algpseudocode}

\newtheorem{theorem}{Theorem}
\newtheorem{observation}{Observation}
\newtheorem{definition}{Definition}

\newcommand{\prob}[1]{\textup{\textsc{\lowercase{#1}}}\xspace}
\newcommand{\Oh}{\mathcal{O}} 
\newcommand{\NP}{\ensuremath{\textsf{NP}}\xspace}
\newcommand{\FPT}{\ensuremath{\textsf{FPT}}\xspace}
\newcommand{\APSP}{\prob{All-Pairs Shortest Paths}}
\newcommand{\TC}{\prob{Triangle Counting}}

\DeclareMathOperator{\cw}{\mathsf{cw}}
\DeclareMathOperator{\mw}{\mathsf{mw}}
\DeclareMathOperator{\tww}{\mathsf{tww}}

\newcommand{\n}[1]{n_{#1}}
\newcommand{\m}[1]{m_{#1}}
\newcommand{\tr}{t}

\newcommand{\NumRedE}[2]{\mu_{#1,#2}}
\newcommand{\TriInGraph}[1]{T(#1)}

\newcommand{\SetOfVertex}[1]{#1(G)}
\newcommand{\GraphOfVertex}[1]{G^{#1}}
\newcommand{\ObsCase}[1]{$(#1)$}

\makeatletter 
\renewcommand\subparagraph{\@startsection{subparagraph}{5}{0pt}% 
                                       {3.25ex \@plus1ex \@minus .2ex}% 
                                       {-1em}% 
                                      {\normalfont\normalsize\bfseries}} 
\makeatother

\title{On Triangle Counting Parameterized by Twin-Width}

\author[1]{Stefan Kratsch}
\author[1]{Florian Nelles}
\author[1]{Alexandre Simon}

\date{}

\affil[1]{\small Department of Computer Science, Humboldt-Universit{\"a}t zu
Berlin, Germany  

{\{kratsch, nelles, simonarm\}@informatik.hu-berlin.de}}

\begin{document}

\maketitle
\abstract{In this report we present an algorithm solving \TC in time $\Oh(d^2n+m)$, where $n$ and $m$, respectively, denote the number of vertices and edges of a graph $G$ and $d$ denotes its twin-width, a recently introduced graph parameter. We assume that a compact representation of a $d$-contraction sequence of $G$ is given.}

\section{Introduction}
\label{sec:introduction}

Recently, the parameter twin-width was introduced by Bonnet et al.~\cite{Bonnet0TW20} over a series of papers.
The $\emph{twin-width}$ of a graph $G$, denoted by $\tww(G)$, roughly measures the distance of a graph from being a cograph, and has demonstrated to have various advantages that make it stand out as an important graph parameter.
One of the main motivations to study twin-width comes from the fact that the class of bounded twin-width contains several interesting and diverse graphs as bounded boolean-width, bounded rank-width, bounded clique-width, unit interval, proper-minor closed  or also $K_t$-minor free graphs~\cite{BonnetG0TW21}. 
On the algorithmic side, FO model checking is \FPT on classes with bounded twin-width~\cite{Bonnet0TW20}. 
Moreover, various intractable problems like \prob{Independent Set}, \prob{Dominating Set}, and \prob{Clique} can be solved in time $2^{\Oh(k)}n$-time \cite{BonnetG0TW21}. 
It is also shown there that graphs of twin-width $d$ admit an interval biclique partition of size $\Oh(dn)$. Using such an edge-partition, they show how to solve \APSP in time $\Oh(n^2\log{n})$.

In this report, we show how to solve the \TC problem on graphs with $n$ vertices and $m$ edges in time $\Oh(d^2n+m)$, with $d$ denoting the twin-width of the graph. 
A graph $G$ has twin-width at most $d$ if there is a so-called $d$-contraction sequence of $G$, which is defined in Section~\ref{sec:preliminaries}. 
Deciding if the twin-width of a graph is at most $4$ is \NP-hard~\cite{BergBD21}. Thus, we assume that a $d$-contraction sequence is given together with the graph $G$.

The currently fastest (unparameterized) algorithm for \TC is due to 
Alon, Yuster, and Zwick~\cite{AlonYZ97}, that have showed that \TC can be solved in time $\Oh(n^\omega)$ using fast matrix multiplication where $2 \leq \omega < 2.372863$. For sparse graphs one can solve \TC in time $\Oh(m^{\frac{2\omega}{\omega+1}}) = \Oh(m^{1.41})$.
It is conjectured that there is no $\Oh(n^{3-\varepsilon})$ time combinatorial algorithm.
Within a parameterized framework, Coudert, Ducoffe, and Popa gave an algorithm that runs in time $\Oh(\cw^2(n+m))$ where $\cw$ denotes the $\emph{clique-width}$ of the input graph \cite{CoudertDP19} whereas Kratsch and Nelles obtained an $\Oh(\mw^{\omega-1}n+m)$ where $\mw$ is the $\emph{modular-width}$ of the input graph~\cite{KratschN18}.

\section{Preliminaries}
\label{sec:preliminaries}

All graphs considered are finite, undirected, and simple. We refer to \cite{0030488} for the basic concepts and notions of graph theory.
In particular, given a graph $G = (V,E)$, we denote by $V(G)$ its vertex set and by $E(G)$ its edge set.
For a graph $G=(V,E)$ and a subset $X \subseteq V$, we define the induced subgraph over the vertex set $X$ as $G[X] = (X , E')$ where $E' = \{\{u,v\} \in V \mid u, v \in X\}$.
We refer by $N(v)$ to the set of neighbors of a vertex $v \in V(G)$, i.e., $N(v) = \{ u \in V(G) \mid \{u,v\} \in E(G) \} $.
Given a vertex $v \in V(G)$, the \emph{degree} $d(v)$ of $v$ is the number of neighbors of $v$, i.e., $d(v) = |N(v)|$.
Furthermore, we say that two vertices $u,v \in V(G)$ are \emph{twins} if $N(v) \setminus \{u, v \} = N(u) \setminus \{u, v\}$.
For any integers $j, k \in \mathbb{N}$, we denote $[j,k] = \{j,j+1,\ldots,k \}$ and in particular $[k]=[1,k]=\{1,2,\ldots,k\}$.

\subparagraph{Twin-width.}
A \emph{trigraph} $G$ is a triple $G = (V, E, R)$ where $E$ and $R$ are two disjoint sets of edges. We refer to an edge in $E$ as a black edge and to an edge in $R$ as a red edge.
By setting $R = \emptyset$, one can interprete any graph $(V,E)$ as a trigraph $(V,E,\emptyset)$. 
A trigraph $G = (V, E, R)$ with maximum red degree $d$, i.e., maximum degree in the graph $(V, R)$, is called a $d$-trigraph.
Furthermore, for any trigraph $G=(V,E,R)$ and any vertex $v \in V$, we denote by $N_R(v)$ the set of red neighbors of $v$, i.e., $N_R(v) = \{ u \in V \mid \{u,v\} \in R \} $.

For a trigraph $G = (V,E,R)$ and two vertices $u, v \in V$, we define $G/u,v = (V', E', R')$ as the trigraph obtained from $G$ by contracting $u$ and $v$ into a new vertex $w$ and after updating the edge sets in the following way: A vertex $x$ is linked to the new vertex $w$ in $G/u,v$ by a black edge if and only if $x$ is linked to $u$ \emph{and} to $v$ in $G$ by a black edge. Moreover, $x$ is not adjacent to $w$, if $x$ is neither adjacent to $u$ nor to $v$ in $G$. In all other cases $x$ is linked to $w$ by a red edge. Formally, $V' = (V \setminus \{u,v\} \cup \{w\})$ with $\{w, x \} \in E'$ if and only if $ \{u,x\} \in E$ and $ \{v,x\} \in E$; $\{w,x\} \notin E' \cup R'$ if and only if $ \{u,x\} \notin E \cup R$ and $\{v,x\} \notin E \cup R$; and $\{w,x\} \in R'$ otherwise. All edges that are not incident to $u$ nor to $v$ remain unchanged in $G/u,v$. Notice that $u$ and $v$ do not need to be adjacent.
For any integer $d \geq 0$, if both $G$ and $G/u,v$ are $d$-trigraphs, $G/u,v$ is called a \emph{$d$-contraction}.
A trigraph $G$ is \emph{$d$-collapsible} if there exists a sequence of $d$-contractions which contracts $G$ to a single vertex.
The minimum integer $d \geq 0$ such that $G$ is $d$-collapsible is called the twin-width of $G$, denoted by $\tww(G)$.
In other words, for any graph $G$ with $\tww(G) = d$, there exists a sequence of trigraphs $G_n, G_{n-1}, \ldots, G_2, G_1$ with $G_n = G$, $G_1 = K_1$ (the clique of size $1$) and $G_k$ is a $d$-contraction of $G_{k+1}$ for $k \in [n-1]$.
To represent such a contraction sequence efficiently, it is sufficient to only specify the vertices that get contracted:
\begin{definition}[Compact representation of a $d$-sequence]\label{def:compactRepresentation}
    Let $G = G_n, G_{n-1}, \ldots, G_1 = K_1$ be a $d$-contraction sequence of an $n$-vertex graph $G = (V,E)$ with $V = \{v_1, v_2, \ldots, v_n\}$. Then, we call $(v_{i_k},v_{j_k})_{n \geq k \geq 2}$ with 
    %$i_{n-k} \in [n+k] \setminus \{i_\ell \mid \ell > k\}$, 
    $G_{k-1} = G_{k}/v_{i_k},v_{j_k}$ a \emph{compact representation} of a $d$-contraction sequence. The graph $G_{k-1}$ results from $G_{k}$ by contracting the two vertices $v_{i_{k}}$ and $v_{j_{k}}$ into a new vertex $v_{{2n-k+1}}$ for $k \in [2,n]$. 
    % and $i_{k} \in [2n-k] \setminus (\{i_\ell \mid n \geq  \ell \geq k\}\cup \{j_\ell \mid n \geq  \ell \geq k\})$.
\end{definition}

Finally, for a vertex $v \in V(G_k)$, we denote by $\SetOfVertex{v}$ the subset of vertices in $G$ eventually contracted into $v$ in $G_k$ and we denote $\GraphOfVertex{v} = G[\SetOfVertex{v}]$.

\section{Algorithm}\label{sec:algorithm}

In the \TC problem, we are given a graph $G = (V,E)$ and we are asked to count the number of triangles in $G$, that is, the number of elements in the set $\TriInGraph{G} = \{\{x,y,z\} \subseteq \binom{V}{3} \mid \{x,y\},\{y,z\},\{z,x\} \in E  \}$. We will prove the following theorem:

\begin{theorem}\label{thm:TriangleCounting}
    Let $G = (V,E)$ be a graph with $\tww(G) = d$, and let a compact representation of a $d$-contraction sequence as defined in Definition~\ref{def:compactRepresentation} be given. Then, one can solve \TC in time $\Oh(d^2n+m)$.
\end{theorem} 

Using the compact representation of the $d$-contraction sequence, we gradually construct the graphs $G = G_n, G_{n-1}, \ldots, G_1 = K_1$. 
Consider a trigraph $G_k = (V_k, E_k, R_k)$ of the contraction sequence of $G$ for $k \in [n]$ and a fixed triangle $\{a,b,c\}$ in $G$ with $a,b,c \in V(G)$. The vertices of the triangle can be in subgraphs corresponding to one, two, or three vertices of $V_k$. More formally, we observe the following:

\begin{observation}\label{obs:CasesOfTrianglesInGi}
Let $G=G_n ,\ldots, G_1 = K_1$ be a contraction sequence of a graph $G$, let $G_k = (V_k, E_k, R_k)$ a trigraph of the contraction sequence, and let $\{a,b,c\}$ be a triangle in $G$. Then, exactly one of the following statements is true (after possibly reordering $a$, $b$, and $c$):
\begin{enumerate}[$(i)$]
	\item $a \in \SetOfVertex{x}$, $b \in \SetOfVertex{y}$, $c \in \SetOfVertex{z}$ with $\{x,y\},\{y,z\},\{z,x\} \in E_k$ \label{obs:item:TriangleThreeBlack}
	
	\item $a \in \SetOfVertex{x}$, $b \in \SetOfVertex{y}$, $c \in \SetOfVertex{z}$ with $\{x,y\},\{y,z\} \in E_k$, $\{z,x\} \in R_k$ \label{obs:item:TriangleTwoBlack}
	
	\item $a \in \SetOfVertex{x}$, $b \in \SetOfVertex{y}$, $c \in \SetOfVertex{z}$ with $\{x,y\} \in E_k$, $\{y,z\},\{z,x\} \in R_k$ \hfill $(\star)$
	\label{obs:item:TriangleOneBlack}
	
	\item $a \in \SetOfVertex{x}$, $b \in \SetOfVertex{y}$, $c \in \SetOfVertex{z}$ with  $\{x,y\},\{y,z\},\{z,x\} \in R_k$ \hfill $(\star)$
	\label{obs:item:TriangleZeroBlack}
	
        \item $a \in \SetOfVertex{x}$ and $b,c \in \SetOfVertex{y}$ with $\{x,y\} \in E_k$ 
        \label{obs:item:EdgeBlack}
		
	\item $a \in \SetOfVertex{x}$ and $b,c \in \SetOfVertex{y}$ with $\{x,y\} \in R_k$ \hfill $(\star)$
	\label{obs:item:EdgeRed}
	
	\item $a,b,c \in \SetOfVertex{x}$ for $x \in V(G_k)$ \hfill $(\star)$
	\label{obs:item:Vertex}
\end{enumerate}
\end{observation}

Let $G_{k-1} = G_k/u,v$.
A triangle $T$ of $G$ might transition from a case in $G_{k-1}$ to a different case in $G_k$ if $T$ consists of vertices in $\SetOfVertex{u}$ or $\SetOfVertex{v}$. If $T$ consists of vertices in both $\SetOfVertex{u}$ and $\SetOfVertex{v}$, one vertex less is needed in $G_{k-1}$ to specify $T$. If $T$ does only admit a non-empty cut with one of the two sets $\SetOfVertex{u}$ or $\SetOfVertex{v}$, the incident edges of $u$ (resp.\ $v$) in $G_{k-1}$ might turned red.
See Figure~\ref{fig:ChangeOfCases} for a full diagram of all possible case transitions for a triangle in $G$ from $G_k$ to $G_{k-1}$.

\begin{figure}[t]
    \centering
    \begin{tabular}{c|c}
        \begin{subfigure}[b]{0.5\textwidth}
            \begin{tikzpicture}
            		\node [gray] (1) at (0, 4) {\ObsCase{\ref*{obs:item:TriangleThreeBlack}}};
            		\node [gray] (2) at (0, 3) {\ObsCase{\ref*{obs:item:TriangleTwoBlack}}};
            		\node [black,very thick] (3) at (0, 2) {\ObsCase{\ref*{obs:item:TriangleOneBlack}}$^*$};
            		\node [black,very thick] (4) at (0, 1) {\ObsCase{\ref*{obs:item:TriangleZeroBlack}}$^*$};
            		\node [gray] (5) at (3, 2) {\ObsCase{\ref*{obs:item:EdgeBlack}}};
            		\node [black,very thick] (6) at (3, 1) {\ObsCase{\ref*{obs:item:EdgeRed}}$^*$};
            		\node [black,very thick] (7) at (6, 1) {\ObsCase{\ref*{obs:item:Vertex}}$^*$};
            
            		\draw [->,gray] (1) to (5);
            		\draw [->,gray] (2) to (5);
            		\draw [->,gray] (3) to (6);
             		\draw [->,gray] (4) to (6);
            		\draw [->, very thick] (5) to (7);
            		\draw [->,gray] (6) to (7);
            		\draw [->,gray] (1) to (2);
            		\draw [->, very thick] (2) to (3);
            		\draw [->,gray] (3) to (4);
            		\draw [->, very thick] (5) to (6);
            		\draw [->, very thick] (2) to (6);
            		\draw [->, bend right= 50, very thick] (1) to (3);
            		\draw [->, bend right = 50, very thick] (2) to (4);
            \end{tikzpicture}
            \caption{Diagram of case transitions.}
            \label{fig:red_square}
        \end{subfigure}
        &
        \begin{subfigure}[b]{0.5\textwidth}
            \resizebox{0.9\columnwidth}{!}{
            \begin{tikzpicture}[]
                    \node[shape=circle,draw=black,fill, inner sep=2pt] (11)  at (0,5) {};
                    \node[shape=circle,draw=black,fill, inner sep=2pt] (12)  at (1,5) {};
                    \node[shape=circle,draw=black,fill, inner sep=2pt] (13)  at (0.5,5.85) {};
                    
                    \draw (11) -- (12);
                    \draw (11) -- (13);
                    \draw (12) -- (13);
            		\node [black] (0) at (0.5, 4.5) {\ObsCase{\ref*{obs:item:TriangleThreeBlack}}};	%-----------------------------------------------------------------------------
            		\node[shape=circle,draw=black,fill, inner sep=2pt] (21)  at (2,5) {};
                    \node[shape=circle,draw=black,fill, inner sep=2pt] (22)  at (3,5) {};
                    \node[shape=circle,draw=black,fill, inner sep=2pt] (23)  at (2.5,5.85) {};
                    
                    \draw[red] (21) -- (22);
                    \draw (21) -- (23);
                    \draw (22) -- (23);
            		\node [black] (1) at (2.5, 4.5) {\ObsCase{\ref*{obs:item:TriangleTwoBlack}}};%---------------------------------------------------------------------------------	
                    \node[shape=circle,draw=black,fill, inner sep=2pt] (31)  at (4,5) {};
                    \node[shape=circle,draw=black,fill, inner sep=2pt] (32)  at (5,5) {};
                    \node[shape=circle,draw=black,fill, inner sep=2pt] (33)  at (4.5,5.85) {};
                    
                    \draw (31) -- (32);
                    \draw[red,] (31) -- (33);
                    \draw[red] (32) -- (33);
            		\node[] (2) at (4.5, 4.5) {\ObsCase{\ref*{obs:item:TriangleOneBlack}}$^*$};%---------------------------------------------------------------------------------
                    \node[shape=circle,draw=black,fill, inner sep=2pt] (41)  at (6,5) {};
                    \node[shape=circle,draw=black,fill, inner sep=2pt] (42)  at (7,5) {};
                    \node[shape=circle,draw=black,fill, inner sep=2pt] (43)  at (6.5,5.85) {};
                    
                    \draw[red] (41) -- (42);
                    \draw[red] (41) -- (43);
                    \draw[red] (42) -- (43);		
            		\node[] (3) at (6.5, 4.5) {\ObsCase{\ref*{obs:item:TriangleZeroBlack}}$^*$};%---------------------------------------------------------------------------------	
            		\node[shape=circle,draw=black,fill, inner sep=2pt] (51)  at (1,3) {};
                    \node[shape=circle,draw=black,fill, inner sep=2pt] (52)  at (2,3) {};
                    
                    \draw (51) -- (52);
            		\node (4) at (1.5, 2.5) {\ObsCase{\ref*{obs:item:EdgeBlack}}};
                    %---------------------------------------------------------------------------------		
            		\node[shape=circle,draw=black,fill, inner sep=2pt] (61)  at (3.5,3) {};
                    \node[shape=circle,draw=black,fill, inner sep=2pt] (62)  at (4.5,3) {};		
            		
            		\draw[red] (61) -- (62);
            		\node (5) at (4, 2.5) {\ObsCase{\ref*{obs:item:EdgeRed}}$^*$};
                    %---------------------------------------------------------------------------------            	
            		\node[shape=circle,draw=black,fill, inner sep=2pt] (71)  at (6,3) {};		
            		\node (6) at (6, 2.5) {\ObsCase{\ref*{obs:item:Vertex}}$^*$};
            \end{tikzpicture}
            }
            \caption{Illustration of the different cases.}
            \label{fig:blue_square}
        \end{subfigure}
    \end{tabular}
    \caption{Possible case transitions of a triangle in $G$ from $G_k$ to $G_{k-1}$ as described in Observation~\ref{obs:CasesOfTrianglesInGi}. Cases that are stored in the variable $t$ in the invariant are marked by a star. Selfloops are omitted. Transitions from an unmarked case to a marked case are represented by thick arrows.}
    \label{fig:ChangeOfCases}
\end{figure}
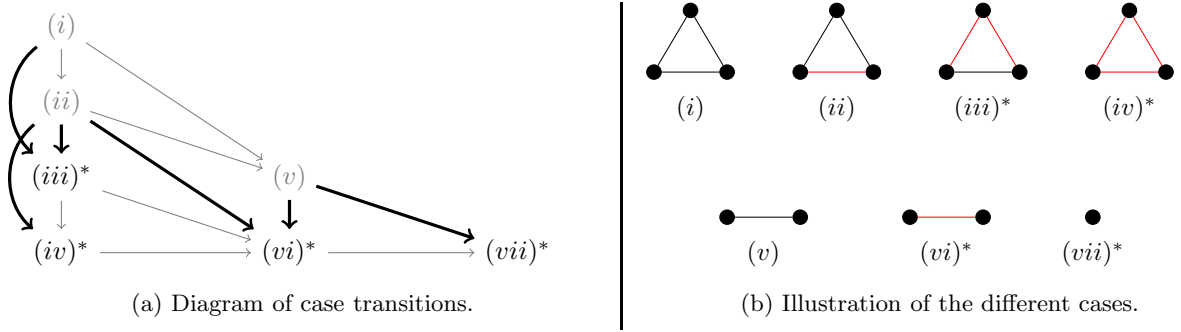    

Over the course of the algorithm, we store the number of triangles of $G$ that appear in $G_k$ as a case marked by a star $(\star)$ in Observation~\ref{obs:CasesOfTrianglesInGi} in a variable $t_k$. To do so, we keep track of the number of vertices and the number of edges of the subgraphs of $G$ that get contracted into a vertex $x$ in $G_k$, i.e., the values $\n{x} = \vert V(\GraphOfVertex{x}) \vert $ and $\m{x} = \vert E(\GraphOfVertex{x}) \vert$ for each vertex $x \in V_k$. Also, we store for each red edge $\{x,y\} \in R_k$, the number of edges between $\GraphOfVertex{x}$ and $\GraphOfVertex{y}$, i.e., $\NumRedE{x}{y} = \vert \{ \{a, b\} \in E \mid a \in \SetOfVertex{x}, b \in  \SetOfVertex{y} \} \vert$.

In each iteration $k$ of the algorithm, the two vertices $v_{i_k}$ and $v_{j_k}$, given in the compact representation of the contraction sequence, are contracted to form the trigraph $G_{k-1}$. The algorithm updates the auxiliary values $\n{x}$, $\m{x}$, and $\NumRedE{x}{y}$ that have changed, for all $x,y \in V_k$.
Informally, the algorithm then counts a triangle of $G$ whenever it becomes a case marked by a star in Observation~\ref{obs:CasesOfTrianglesInGi} for the first time. Whenever a triangle is of a marked case, it cannot transition back to an unmarked case and, eventually, all triangles will be of Case~\ObsCase{\ref{obs:item:Vertex}}. Note, that for $G = G_n$, every triangle is of Case~\ObsCase{\ref{obs:item:TriangleThreeBlack}}.

More precisely, the transition from Case~\ObsCase{\ref{obs:item:EdgeBlack}} to Case~\ObsCase{\ref{obs:item:Vertex}} is dealt by the main Algorithm \ref{alg:BaseAlgorithm}. In addition, the procedure \textsc{TriCountOneNeighbor} focuses on the transitions from Case~\ObsCase{\ref{obs:item:TriangleTwoBlack}} and Case~\ObsCase{\ref{obs:item:EdgeBlack}} to Case~\ObsCase{\ref{obs:item:EdgeRed}} whereas the procedure \textsc{TriCountTwoNeighbors} handles the transitions from Case~\ObsCase{\ref{obs:item:TriangleThreeBlack}} and Case~\ObsCase{\ref{obs:item:TriangleTwoBlack}} to Case~\ObsCase{\ref{obs:item:TriangleOneBlack}} and from Case~\ObsCase{\ref{obs:item:TriangleTwoBlack}} to Case~\ObsCase{\ref{obs:item:TriangleZeroBlack}}.

Consider a trigraph $G_k$ that will be contracted into $G_{k-1} = G_k/{v_{i_k},v_{j_k}}$ according to the contraction sequence. For simplicity, we define $u := v_{i_k}$  and $v := v_{j_k}$ as the two vertices in $V(G_k)$ that get contracted into the new vertex $w := v_{2n-k+1}$ of $G_{k-1}$.

For the vertex $w \in V(G_{k-1})$, the number of vertices in $\GraphOfVertex{w}$ is the sum of these numbers in $\GraphOfVertex{u}$ and $\GraphOfVertex{v}$. For the number of edges, we also need to add the number of edges between $\GraphOfVertex{u}$ and $\GraphOfVertex{v}$. Therefore, $\n{w} = \n{u} + \n{v}$ and 
\begin{align*}
    \m{w} = \begin{cases} 
    \m{u} + \m{v} + \n{u} \cdot \n{v}  & \mbox{if $\{u,v\} \in E_k$,} \\ 
    \m{u} + \m{v} + \NumRedE{u}{v}  & \mbox{if $\{u,v\} \in R_k$,} \\ 
    \m{u} + \m{v} & \mbox{otherwise.} 
          \end{cases}
\end{align*}
For every other vertex $x \in V(G_{k-1})$, $x \neq w$, these values remain unchanged.
Finally, for any vertex $x \in V(G_{k-1})$, such that $\{w,x\} \in R(G_{k-1})$ the number of edges between $G^w$ and $G^x$ can be computed as follows: 
\begin{align*}
    \NumRedE{w}{x} = 
    \begin{cases} 
        \n{u} \cdot \n{x}  & \mbox{if $\{u,x\} \in E_k$ and $\{v,x\} \notin (E_k \cup R_k)$,} \\ 
        \n{v} \cdot \n{x}  & \mbox{if $\{u,x\} \notin (E_k \cup R_k)$ and $\{v,x\} \in E_k$,} \\ 
        \NumRedE{u}{x} & \mbox{if $\{u,x\} \in R_k$ and $\{v,x\} \notin (E_k \cup R_k)$,} \\ 
        \NumRedE{v}{x} & \mbox{if $\{u,x\} \notin (E_k \cup R_k)$ and $\{v,x\} \in R_k$,} \\
        \NumRedE{u}{x} + \n{v} \cdot \n{x} & \mbox{if $\{u,x\} \in R_k$ and $\{v,x\} \in E_k$,} \\ 
        \NumRedE{v}{x} + \n{u} \cdot \n{x} & \mbox{if $\{u,x\} \in E_k$ and $\{v,x\} \in R_k$,} \\ 
        \NumRedE{u}{x} + \NumRedE{v}{x}  & \mbox{if $\{u,v\} \in R_k$ and $\{v,x\} \in R_k$.} 
    \end{cases}
\end{align*}
For every other vertex $y \in V(G_{k-1})$, $y \neq w$, such that $ \{x,y\} \in R(G_{k-1})$, the value $\NumRedE{x}{y}$ remains unchanged.

We give a pseudocode of the algorithm below. For algorithmic purposes, we assume that we are given a graph $G=(V,E)$ that will, over the course of the algorithm, be updated into the successive trigraphs defined by the contraction sequence. Similarly, the variable $t$ in the algorithm represents the number of triangles in $G$ computed so far. The procedure \textsc{UpdateAuxiliaryValues} is not given but is explained previously whereas the procedures \textsc{TriCountOneNeighbor} and \textsc{TriCountTwoNeighbors} will be described in the next paragraphs. Finally, \textsc{UpdateGraph} performs the actual contraction of the graph $G$. The pseudocode of this procedure is omitted.

\begin{algorithm}
    \hspace*{\algorithmicindent} \textbf{Input:}  A graph $G=(V,E)$ and a compact representation of a contraction sequence $(v_{i_k},v_{j_k})_{n \geq k \geq 2}$\\
    \hspace*{\algorithmicindent} \textbf{Output: }  The number of triangles $\tr$ in the graph $G$
    \begin{algorithmic}[1]
        \State $\tr := 0$ \Comment{number of triangles in $G$ of Case~\ObsCase{\ref{obs:item:TriangleOneBlack}}, \ObsCase{\ref{obs:item:TriangleZeroBlack}}, \ObsCase{\ref{obs:item:EdgeRed}}, and \ObsCase{\ref{obs:item:Vertex}}}
        \State $R = \emptyset$
        \For{every vertex $x \in V$}                               
            \State $\n{x} := 1$       \Comment{number of vertices in $\GraphOfVertex{x}$}   
            \State $\m{x} := 0$       \Comment{number of edges in $\GraphOfVertex{x}$}
        \EndFor
        \For{$n \geq k \geq 2$}                              
            \State $u := v_{i_k}$
            \State $v := v_{j_k}$
            \State $w := v_{2n-k+1}$
            \State \textsc{UpdateAuxiliaryValues}$(G,u,v,w)$ \;
            \If{$\{u,v\} \in E$}
                \State $\tr = \tr + \n{u} \cdot \m{v} + \n{v} \cdot \m{u}$
                \Comment{Case~\ObsCase{\ref{obs:item:EdgeBlack}} $\rightarrow$ Case~\ObsCase{\ref{obs:item:Vertex}}}\label{alg:BlackEdge:Vertex}
            \EndIf
            \For{each $x \in N_R(w)$} 
                \State \textsc{TriCountOneNeighbor}$(G,u,v,w,x)$        
                \State \textsc{TriCountTwoNeighbors}$(G,u,v,w,x)$        
            \EndFor
                \State \textsc{UpdateGraph}$(G,u,v,w)$  
        \EndFor
        \State \textbf{return} $t$
    \end{algorithmic}
    \caption{}
    \label{alg:BaseAlgorithm}
\end{algorithm}

In the procedure \textsc{TriCountOneNeighbor}, we focus on the red edge between the newly introduced vertex $w$ and one of its red neighbors $x \in N_R(w)$. We then consider the different edges between $u$, $v$, and $x$ to detect triangles in $G$ that transition from Case~\ObsCase{\ref{obs:item:TriangleTwoBlack}} and Case~\ObsCase{\ref{obs:item:EdgeBlack}} to Case~\ObsCase{\ref{obs:item:EdgeRed}}.

\begin{algorithm}[H]
    \begin{algorithmic}[1]
        \Procedure{TriCountOneNeighbor}{$G,u,v,w,x$}
            \If{$\{u,x\} \in E$}
                \State $\tr = \tr + \n{u} \cdot \m{x} + \n{x} \cdot \m{u}$ 
                \Comment{Case~\ObsCase{\ref{obs:item:EdgeBlack}} $\rightarrow$ Case~\ObsCase{\ref{obs:item:EdgeRed}}}\label{alg:EdgeBlack:EdgeRed:u}
                    \If{$\{u,v\}\in E$ and $\{v,x\} \in R$}
                        \State $\tr = \tr + \NumRedE{v}{x} \cdot \n{u}$
                        \Comment{Case~\ObsCase{\ref{obs:item:TriangleTwoBlack}} $\rightarrow$ Case~\ObsCase{\ref{obs:item:EdgeRed}}}\label{alg:TriangleTwoBlack:EdgeRed:2}
                    \EndIf
            \ElsIf{$\{v,x\} \in E$}
                        \State $\tr = \tr + \n{v} \cdot \m{x} + \n{x} \cdot \m{v}$
                        \Comment{Case~\ObsCase{\ref{obs:item:EdgeBlack}} $\rightarrow$ Case~\ObsCase{\ref{obs:item:EdgeRed}}}\label{alg:EdgeBlack:EdgeRed:v}
                        \If{$\{u,v\}\in E$ and $\{u,x\} \in R$}
                            \State $\tr = \tr + \NumRedE{u}{x} \cdot \n{v}$
                            \Comment{Case~\ObsCase{\ref{obs:item:TriangleTwoBlack}} $\rightarrow$ Case~\ObsCase{\ref{obs:item:EdgeRed}}}\label{alg:TriangleTwoBlack:EdgeRed:1}
                        \EndIf  
            \EndIf
        \EndProcedure
    \end{algorithmic}
    \label{alg:TriCountOverRedEdge}
\end{algorithm}

Finally, in the procedure \textsc{TriCountTwoNeighbors} we focus on the edges between the newly introduced vertex $w$ and two of its neighbors $x,y \in N_R(w)$. We  then consider the different edges between $u$, $v$, $x$, and $y$ to detect triangles in $G$ that transitions from Case~\ObsCase{\ref{obs:item:TriangleThreeBlack}} and Case~\ObsCase{\ref{obs:item:TriangleTwoBlack}} to Case~\ObsCase{\ref{obs:item:TriangleOneBlack}} and from Case~\ObsCase{\ref{obs:item:TriangleTwoBlack}} to Case~\ObsCase{\ref{obs:item:TriangleZeroBlack}}.

\begin{algorithm}[H]
    \begin{algorithmic}[1]
        \Procedure{TriCountTwoNeighbors}{$G,u,v,w,x$}
            \For{each $y \in V$ such that $\{x,y\} \in R$ and $\{w,y\} \in E$}
                \If{$\{u,x\} \in E$}
                    \State $\tr = \tr + \NumRedE{x}{y} \cdot \n{u}$\Comment{Case~\ObsCase{\ref{obs:item:TriangleTwoBlack}} $\rightarrow$ Case~\ObsCase{\ref{obs:item:TriangleOneBlack}}} \label{alg:TwoBlack:OneBlack:xyRed:u}
                \ElsIf{$\{v,x\} \in E$}
                    \State $\tr = \tr + \NumRedE{x}{y} \cdot \n{v}$
                    \Comment{Case~\ObsCase{\ref{obs:item:TriangleTwoBlack}} $\rightarrow$ Case~\ObsCase{\ref{obs:item:TriangleOneBlack}}}\label{alg:TwoBlack:OneBlack:xyRed:v}
                \EndIf
            \EndFor
            \For{each $y \in V$ such that $\{w,y\} \in R$}
                \If{$\{x,y\} \in E$}
                    \If{$\{u,x\} \in E$ and $\{u,y\} \in E$}
                        \State $\tr = \tr + \n{u} \cdot \n{x} \cdot \n{y}$
                        \Comment{Case~\ObsCase{\ref{obs:item:TriangleThreeBlack}} $\rightarrow$ Case~\ObsCase{\ref{obs:item:TriangleOneBlack}}}\label{alg:ThreeBlack:OneBlack1}
                    \ElsIf{$\{v,x\} \in E$ and $\{v,y\} \in E$}
                        \State $\tr = \tr + \n{v} \cdot \n{x} \cdot \n{y}$
                       \Comment{Case~\ObsCase{\ref{obs:item:TriangleThreeBlack}} $\rightarrow$ Case~\ObsCase{\ref{obs:item:TriangleOneBlack}}}\label{alg:ThreeBlack:OneBlack2}
                    \EndIf
                    \If{$\{u,x\} \in R$ and $\{u,y\} \in E$}
                        \State $\tr = \tr + \NumRedE{u}{x} \cdot \n{y}$
                        \Comment{Case~\ObsCase{\ref{obs:item:TriangleTwoBlack}} $\rightarrow$ Case~\ObsCase{\ref{obs:item:TriangleOneBlack}}}\label{alg:TriangleTwoBlack:TriangleOneBlack:xyBlack:1}
                    \ElsIf{$\{v,x\} \in R$ and $\{v,y\} \in E$}
                        \State $\tr = \tr + \NumRedE{v}{x} \cdot \n{y}$ 
                        \Comment{Case~\ObsCase{\ref{obs:item:TriangleTwoBlack}} $\rightarrow$ Case~\ObsCase{\ref{obs:item:TriangleOneBlack}}}\label{alg:TriangleTwoBlack:TriangleOneBlack:xyBlack:2}
                    \EndIf
                    \If{$\{u,x\} \in E$ and $\{u,x\} \in R$}
                        \State $\tr = \tr + \NumRedE{u}{x} \cdot \n{x}$ 
                        \Comment{Case~\ObsCase{\ref{obs:item:TriangleTwoBlack}} $\rightarrow$ Case~\ObsCase{\ref{obs:item:TriangleOneBlack}}}
                    \ElsIf{$\{v,x\} \in E$ and $\{v,y\} \in R$}\label{alg:TriangleTwoBlack:TriangleOneBlack:xyBlack:3}
                        \State $\tr = \tr + \NumRedE{v}{x} \cdot \n{x}$ 
                        \Comment{Case~\ObsCase{\ref{obs:item:TriangleTwoBlack}} $\rightarrow$ Case~\ObsCase{\ref{obs:item:TriangleOneBlack}}}\label{alg:TriangleTwoBlack:TriangleOneBlack:xyBlack:4}
                    \EndIf
                \ElsIf{$\{x,y\} \in R$}
                    \If{$\{u,x\} \in E$ and $\{u,y\} \in E$}
                        \State $\tr =   \tr +  \NumRedE{x}{y} \cdot \n{u}$
                        \Comment{Case~\ObsCase{\ref{obs:item:TriangleTwoBlack}} $\rightarrow$ Case~\ObsCase{\ref{obs:item:TriangleZeroBlack}}}\label{alg:TriangleTwoBlack:TriangleZeroBlack:1}
                    \ElsIf{$\{v,x\} \in E$ and $\{v,y\} \in E$}
                        \State $\tr =  \tr +  \NumRedE{x}{y} \cdot \n{v}$
                        \Comment{Case~\ObsCase{\ref{obs:item:TriangleTwoBlack}} $\rightarrow$ Case~\ObsCase{\ref{obs:item:TriangleZeroBlack}}}\label{alg:TriangleTwoBlack:TriangleZeroBlack:2}
                    \EndIf
                \EndIf 
            \EndFor
        \EndProcedure
    \end{algorithmic}
    \label{alg:TriCountOverRedWedge}
\end{algorithm}

\noindent
We have now described the algorithm and can prove Theorem~\ref{thm:TriangleCounting}.

\begin{proof}[Proof of Theorem~\ref{thm:TriangleCounting}]
Given the compact representation of the $d$-sequence $(v_{i_k},v_{j_k})_{n \geq k \geq 2}$, the algorithm generates iteratively the contraction sequence $G = G_n, G_{n-1}, \ldots, G_1 = K_1$ with $G_{k-1} = G_k/{v_{i_k},v_{j_k}}$ using the procedure \textsc{UpdateGraph} at the end of each iteration.
The values $\n{x}$, $\m{x}$, and $\NumRedE{x}{y}$ are updated by the procedure \textsc{UpdateAuxiliaryValues} in each iteration as described in the previous paragraph.

To prove that the final value of $t$ is equal to the number of triangles in $G$, we will prove that the following invariant is true at the beginning of each iteration, i.e., for each graph $G_k = (V_k,E_k,R_k)$ in the contraction sequence for $k \in [n]$: 
\begin{align*}
  \vert \TriInGraph{G} \vert = t_k + \underbrace{\sum_{\substack{\{x,y\},\\ \{y,z\},\{x,z\} \in E_k}} \n{x} \n{y} \n{z}}_{\text{Case~\ObsCase{\ref*{obs:item:TriangleThreeBlack}}}} + \underbrace{\sum_{\substack{ \{x,z\} \in R_k \\\{x,y\},\{y,z\}\in E_k }} \NumRedE{x}{z} \n{y}}_{\text{Case~\ObsCase{\ref*{obs:item:TriangleTwoBlack}}}} +\underbrace{\sum_{\{x,y\} \in E_k} (\n{x} \m{y} + \m{x} \n{y})}_{\text{Case~\ObsCase{\ref*{obs:item:EdgeBlack}}}} 
\end{align*}
We denote by $t_k$ the current value of $t$ at the start of iteration $k$ (and $t_1$ the final value after iteration $k = 2$).
Recall that $\vert T(G) \vert$ denotes the number of triangles in $G$.
For $k = n$, the value of $t_n$ is initialized to zero 
since $R_n = \emptyset$, $\m{x} = 0$, and $\n{x}=1$ for all $x \in V_n$. Therefore, the invariant simplifies to the second summand only, which is indeed the desired number of all triangles in $G$.
We will show that the value of the invariant will never change. Thus, for $i = 1$, it then holds that $\vert \TriInGraph{G} \vert = t_1 + 0 + 0+0$ and the correctness of Algorithm~\ref{alg:BaseAlgorithm} follows.

As described in Observation~\ref{obs:CasesOfTrianglesInGi}, we distinguish seven cases of a possible occurrence of a triangle of $G$ in $G_k$. In the beginning, all triangles in $G$ are of Case~\ObsCase{\ref{obs:item:TriangleThreeBlack}} but some may change from a case to another one whenever $G_k$ gets contracted to $G_{k-1}$. For a fixed triangle, all possible case transitions are depicted in Figure~\ref{fig:ChangeOfCases}.
Notice that the triangles of $G$ of Case~\ObsCase{\ref{obs:item:TriangleThreeBlack}}, \ObsCase{\ref{obs:item:TriangleTwoBlack}}, or \ObsCase{\ref{obs:item:EdgeBlack}}, are counted directly by the corresponding sums in the invariant.
We are left to show that the (current) value of $t_k$ is indeed the count of all triangles of $G$ that appear in $G_k$ of Case~\ObsCase{\ref{obs:item:TriangleOneBlack}},\ObsCase{\ref{obs:item:TriangleZeroBlack}}, \ObsCase{\ref{obs:item:EdgeRed}}, or \ObsCase{\ref{obs:item:Vertex}}. 
Notice that once a fixed triangle is of one of the latter cases, this triangle can never transition back to an unmarked case.
 
By induction, we can assume that the invariant is true for $G_k$. To prove the invariant for $k-1$, we keep track of all triangles whose case changes from $G_k$ to $G_{k-1}$ regarding Observation~\ref{obs:CasesOfTrianglesInGi}. Note that we only need to consider the triangles that are of a case that is not marked by a star in $G_k$, but in a case that is marked in $G_{k-1}$.
Let $G_{k-1} = G_k /u,v$ and let $w$ be the new vertex of $G_{k-1}$.

Case~\ObsCase{\ref{obs:item:TriangleThreeBlack}} to Case~\ObsCase{\ref{obs:item:TriangleOneBlack}}: 
Let $\{a,b,c\}$ be a triangle in $G$ that is of Case~\ObsCase{\ref{obs:item:TriangleThreeBlack}} in $G_k$ but of Case~\ObsCase{\ref{obs:item:TriangleOneBlack}} in $G_{k-1}$. 
This implies that there exist $x,y \in V_{k-1}$ with $a \in \SetOfVertex{w}$, $b \in \SetOfVertex{x}$, and $c \in \SetOfVertex{y}$ 
and $\{w,x\}, \{w,y\} \in R_k$. Since $w$ is the contraction of $u$ and $v$, it holds that either $a \in \SetOfVertex{u}$ with $\{u,x\}, \{u,y\} \in R_k$ or $a \in \SetOfVertex{v}$ with $\{v,x\}, \{v,y\} \in R_k$. In the former case, it is counted in the procedure \textsc{TriCountTwoNeighbors}, Line~\ref{alg:ThreeBlack:OneBlack1}. In the latter case it is counted in Line~\ref{alg:ThreeBlack:OneBlack2}.

Case~\ObsCase{\ref{obs:item:TriangleTwoBlack}} to Case~\ObsCase{\ref{obs:item:TriangleOneBlack}}:
This implies that there exist $x,y \in V_{k-1}$ with $a \in \SetOfVertex{w}$, $b \in \SetOfVertex{x}$, and $c \in \SetOfVertex{y}$.
Let us first assume that $\{x,y\} \in R_{k-1}$. Since $\{a,b,c\}$ is a triangle of Case~\ObsCase{\ref{obs:item:TriangleTwoBlack}} in $G_k$, it holds that either $\{u,x\},\{u,y\} \in E_k$ or $\{v,x\},\{x,y\} \in E_k$. In the former case, it is counted in the procedure \textsc{TriCountTwoNeighbors}, Line~\ref{alg:TwoBlack:OneBlack:xyRed:u}, and in the latter case, in Line~\ref{alg:TwoBlack:OneBlack:xyRed:v}.
Now assume that $\{x,y\} \in E_{k-1}$, i.e., $\{w,x\},\{w,y\} \in R_{k-1}$.  Since $\{a,b,c\}$ is a triangle of Case~\ObsCase{\ref{obs:item:TriangleTwoBlack}} in $G_k$, it now holds that either  $\{u,x\} \in R_k$ and $\{u,y\} \in E_k$ (counted in the procedure \textsc{TriCountTwoNeighbors}, Line~\ref{alg:TriangleTwoBlack:TriangleOneBlack:xyBlack:1}); $\{v,x\} \in R_k$ and $\{v,y\} \in E_k$ (Line~\ref{alg:TriangleTwoBlack:TriangleOneBlack:xyBlack:2}); $\{u,x\} \in E_k$ and $\{u,y\} \in R_k$ (Line~\ref{alg:TriangleTwoBlack:TriangleOneBlack:xyBlack:3}); or $\{v,x\} \in E_k$ and $\{v,y\} \in R_k$ (Line~\ref{alg:TriangleTwoBlack:TriangleOneBlack:xyBlack:4}).
Note that since $\{w,x\}$, $\{w,y\} \in R_{k-1}$, it cannot be that the first two or the last two cases occur simultaneously.

Case~\ObsCase{\ref{obs:item:TriangleTwoBlack}} to  Case~\ObsCase{\ref{obs:item:TriangleZeroBlack}}:
Suppose there are $x,y \in V_{k-1}$ with $a \in \SetOfVertex{w}$, $b \in \SetOfVertex{x}$, $c \in \SetOfVertex{y}$ and $\{w,x\}, \{w,y\}, \{x,y\}\in R_{k-1}$. Since $\{a,b,c\}$ is of Case~\ObsCase{\ref{obs:item:TriangleTwoBlack}} in $G_k$, it either holds that $\{u,x\}, \{u,y\} \in E_k$ or $\{v,x\}, \{v,y\} \in E_k$. The former is counted in the procedure \textsc{TriCountTwoNeighbors}, Line~\ref{alg:TriangleTwoBlack:TriangleZeroBlack:1} and the latter in Line~\ref{alg:TriangleTwoBlack:TriangleZeroBlack:2}.
    
Case~\ObsCase{\ref{obs:item:TriangleTwoBlack}} to Case~\ObsCase{\ref{obs:item:EdgeRed}} and Case~\ObsCase{\ref{obs:item:EdgeBlack}} to Case~\ObsCase{\ref{obs:item:EdgeRed}}:
Let $\{a,b,c\}$ be a triangle of Case~\ObsCase{\ref{obs:item:EdgeRed}} in $G_{k-1}$, i.e., there exists $x \in V_{k-1}$ with $\{w,x\} \in R_{k-1}$ and either $a \in \SetOfVertex{w}$ and $b,c \in \SetOfVertex{x}$ or $a,b \in \SetOfVertex{w}$ and $c \in \SetOfVertex{x}$. If $\{a,b,c\}$ is of Case~\ObsCase{\ref{obs:item:TriangleTwoBlack}} in $G_{k-1}$, it holds that $\{u,v\} \in E_k$ (otherwise the edge $\{w,x\}$ would not be in $R_{k-1}$) and that either $\{u,x\} \in R_k$ and $\{v,x\} \in E_k$ (counted in the procedure \textsc{TriCountOneNeighbor}, Line~\ref{alg:TriangleTwoBlack:EdgeRed:1}), or $\{v,x\} \in R_k$ and $\{u,x\} \in E_k$  (Line~\ref{alg:TriangleTwoBlack:EdgeRed:2}).
If $\{a,b,c\}$ is of Case~\ObsCase{\ref{obs:item:EdgeBlack}} in $G_k$, the black edge is either incident to $u$ (counted in the procedure \textsc{TriCountOneNeighbor}, Line~\ref{alg:EdgeBlack:EdgeRed:u}) or to $v$ (Line~\ref{alg:EdgeBlack:EdgeRed:v}).

Case~\ObsCase{\ref{obs:item:EdgeBlack}} to Case~\ObsCase{\ref{obs:item:Vertex}}:
Finally, if a triangle $\{a,b,c\}$ is of Case~\ObsCase{\ref{obs:item:EdgeBlack}} in $G_k$ and of Case~\ObsCase{\ref{obs:item:Vertex}} in $G_{k-1}$ it now holds that $a,b,c \in \SetOfVertex{w}$ and such a triangle is counted in Algorithm~\ref{alg:BaseAlgorithm}, Line~\ref{alg:BlackEdge:Vertex}.

Thus, the number of all the triangles of $G$ that are of Case~\ObsCase{\ref{obs:item:TriangleOneBlack}}, \ObsCase{\ref{obs:item:TriangleZeroBlack}}, \ObsCase{\ref{obs:item:EdgeRed}}, and  \ObsCase{\ref{obs:item:Vertex}} in $G_{k-1}$ is indeed computed and stored in the variable $t$ after the iteration $k$.
Since the algorithm only increases $t$ whenever a triangle transitions from an unmarked case to a marked case after contraction, the value $t$ is exactly the desired value.

We store the graph in sorted adjacency lists, which can be  initially realized in time $\Oh(n+m)$ using a linear-time sorting algorithm to sort the vertices $v_1, \ldots, v_n$. To contract the two vertices $u$ and $v$ in each iteration, we can scan the sorted adjacency lists of $u$ and $v$ to identify the red neighborhood and black neighborhood of $w$. Since $w$ has at most $d$ incident red edges and, for each black neighbor, we decrease the number of total edges by one, the total running time, for every call of the procedure $\textsc{UpdateGraph}$, sums up to $\Oh(dn+m)$. Since the auxiliary values only change for $w$ and for the incident red edges of $w$, they can be updated in time $\Oh(d)$ per iteration. Eventually, it takes $\Oh(dn)$ for every call of the procedure \textsc{UpdateAuxiliaryValues}. Finally, the procedures $\textsc{TriCountOneNeighbor}$ and $\textsc{TriCountTwoNeighbors}$ are called at most $d$ times per iteration, taking respectively $\Oh(1)$ and\ $\Oh(d)$ time. Thus, the overall running time of Algorithm~\ref{alg:BaseAlgorithm} is $\Oh(d^2n + m)$.
\end{proof}

\section{Conclusion}

We have obtained an efficient parameterized algorithm for \TC parameterized by the twin-width $\tww$ of the input graph. As a matter of fact, the algorithm is adaptive as it runs in time $\Oh(\tww^2n+m)$ whereas the best unparameterized combinatorial algorithms run in time $\mathcal{O}(n^3)$. Our algorithm is based on dynamic programming and stores a few values that need to be updated at each contraction step.

Some future directions would be to extend this approach and design efficient algorithms to solve other tractable problems when parameterized by the twinwidth of the input graph. It would, furthermore, be highly interesting to find the most general parameter for which an adaptive algorithm for \TC exists. Finally, our algorithm is adaptive when compared to the best combinatorial algorithms, however, there exists a non-combinatorial algorithm that runs in time $\Oh(n^{\omega})$ where $\omega < 2.372863$ \cite{AlonYZ97}. An improvement would then be to obtain an $\Oh(\tww^{\omega-1}n + m)$-time algorithm.

\bibliography{main}

\end{document}